\documentclass[journal]{IEEEtran}
\ifCLASSINFOpdf
\else
\fi
\usepackage{graphicx}          
\usepackage{tikz} 
\usetikzlibrary{matrix,arrows,calc}
\usepackage{amssymb}
\usepackage{bm}
\usepackage[cmex10]{amsmath}
\usepackage{mathrsfs}
\usepackage{color}
\usepackage{xcolor}
\usepackage{cite}
\usepackage[ruled,vlined]{algorithm2e}
\usepackage{graphicx}
\usepackage{float}
\usepackage{subfigure}

\usepgflibrary{shapes.symbols}
\usepgflibrary[shapes.symbols]
\usetikzlibrary{shapes.symbols}
\usetikzlibrary[shapes.symbols]

\newtheorem{theoremx}{Theorem}
\newtheorem{proof}{Proof}
\newtheorem{lemmax}{Lemma}

\newtheorem{remark}{Remark}

\newtheorem{assumption}{Assumption}
\newtheorem{problem}{Problem}
\newtheorem{example}{Example}

\begin{document}
%
\title{Prespecified-time observer-based distributed control of battery energy storage systems}


\author{\IEEEauthorblockN{Wu Yang, Shu-Ming Liang, Yan-Wu Wang, and Zhi-Wei Liu  }\\

\thanks{This work is supported by the National Natural Science Foundation of China under Grant 61903147 and the Interdisciplinary Scientific Research Foundation of GuangXi University under Grant 2022JCC023.

Wu Yang is with the Key Laboratory of Disaster Prevention and Structural Safety of Ministry of Education and School of Electrical Engineering, Guangxi University, Nanning, 530004, P.R. China. (e-mail: biwuyang@163.com).

Shu-Ming Liang is with the School of Electrical Engineering, Guangxi University, Nanning, 530004, P.R. China. (e-mail: a1017261884@163.com).

Yan-Wu Wang and Zhi-Wei Liu  are with the School of Artificial Intelligence and Automation, Huazhong University of Science and Technology, Wuhan, 430074, China. (e-mail: wangyw@hust.edu.cn; zwliu@hust.edu.cn).

}}


\IEEEtitleabstractindextext{%
\begin{abstract}
  This paper studies the state-of-charge (SoC) balancing and the total charging/discharging power tracking issues for battery energy storage systems (BESSs) with multiple distributed heterogeneous battery units. Different from the traditional cooperative control strategies based on the asymptotical or finite-time distributed observers, two distributed prespecified-time observers are proposed to estimate average battery units state and average desired power, respectively, which can be determined in advance and independent of initial states or control parameters. Finally, two simulation examples are given to verify the effectiveness and superiority of the proposed control strategy.
\end{abstract}

}

\maketitle

\IEEEdisplaynontitleabstractindextext

%
\IEEEpeerreviewmaketitle

\section{Introduction}

Energy storage system is an indispensable part of microgrid, which can guarantee the power quality and reliability and reduce the energy lose\cite{chen2009progress}.
Among all kinds of energy storage technologies such as supercapacitor, superconducting magnetic, etc, BESS has an irreplaceable position in microgrid because of its advantages of fast response speed, high energy density, high efficiency and flexible configuration\cite{lawder2014battery,hill2012battery}. During past decades, various types of BESSs are rapidly integrated into microgrid\cite{feehally2016battery,gundogdu2018battery}. Despite their continuous advances in electrochemical technology, the control of BESSs remains a very challenging problem\cite{rahimi2013battery}.

Generally, BESSs are composed of multiple battery units, which can communicate with its nearby battery units and monitor and control the charging/discharging power of itself. The primary goal of BESS management is to balance the SoC of all battery units and to meet the total charging/discharging desired power of microgrid. Note that, due to variations and deviations in manufacturing processes and operating conditions, battery units can manifest different characteristics even with the same specifications. Consequently, how to design an appropriate control strategies for BESS has attracted the broad interest of researchers \cite{cai2016distributed,xing2019distributed,deng2020distributed,khazaei2018distributed,meng2021distributed}.

The control methods for BESS are typically divided into centralized, decentralized or distributed ones: in centralized control method, a central controller is needed to monitor and coordinate the SoC and other critical states of all battery units\cite{cao2008battery}, which is costly and prone to single-point failure; Compared with centralized control method, neither central controller nor communication between each battery units is required in decentralized control method, however, it has been reported in \cite{zeng2021hierarchical} that such control method may lead to slow SoC balancing, low SoC balancing accuracy and poor bus voltage quality; By contrast, in distributed control method, no extra centralized controller is needed and each battery unit can communicate with its neighbors, such method not only can save communication cost, but also has the advantages of robustness and reconfigurability \cite{yazdanian2014distributed}. Motivated by this, different distributed cooperative control schemes are developed for BESSs. For examples, \cite{xing2019distributed} proposed a distributed SoC balancing control method based on event-triggered mechanism, which allows each battery unit transmits signals to its neighboring ones only when the triggered condition is met. Unfortunately, the average desired power and average battery unit state are needed when implementing the proposed controller, which is unrealistic since both the average desired power and average battery unit state are of the global information rather than the local one. To overcome drawback in \cite{xing2019distributed}, recently, \cite{meng2021distributed} introduced two types of distributed observers, one can achieve asymptotic estimation on the average desired power and average battery unit state, while the other ensures finite-time estimation. It should be emphasized that the proposed distributed asymptotic observers in \cite{meng2021distributed} can only be carried out as time tends to infinity. In fact, finite-time convergence is more desirable in applications due to the advantage of faster convergence rate and robustness against uncertainties. Although the distributed finite-time observers are further constructed in \cite{meng2021distributed}, the observer convergence rate depends on the initial values and certain parameters of BESSs, which may restrict its applications since the settling time is not fixed for different initial values. Note that in the framework of prespecified-time stability, the settling time can be bounded by a fixed value, independent of the initial conditions and can be pre-set according to the task requirements, which makes prespecified-time stability more desirable. To the best of our knowledge, the SoC balancing and the total charging/discharging power tracking issues for heterogeneous BESSs has not been fully investigated.

Inspired by the above observations, this paper addresses the problem of the SoC balancing and the total charging/discharging power tracking for heterogeneous BESS. The main contribution of this paper is highlighted as follows: Different from the existing control strategies \cite{xing2019distributed,meng2021distributed}, the prespecified-time distributed observer-based cooperative control method is employed to achieve both the SoC balancing and the total charging/discharging power tracking for heterogeneous BESS. Specifically, by introducing time-dependent function, two prespecified-time distributed observers are proposed to reconstruct the average desired power and the average battery unit state, respectively. It is shown that the settling time of the proposed observer can be pre-set by user. Moreover, simulation examples are given to verify that the performance of the proposed control method is better than that of \cite{xing2019distributed,meng2021distributed}.

The rest of the paper is organized as follows. Section II reviews graph theory and describes BESSs. Section III introduces the charging/discharging power controller based on the prespecified-time average power observer and the prespecified-time average battery unit state observer. Section IV provides simulation results. Section V concludes this paper.


\section{Preliminaries and system description}
\subsection{Preliminaries}
In this  subsection, graph theory is employed to describe the communication of BESSs, in which vertex $i\in \mathcal{V}$ refers to battery unit $i\in \mathcal{V}$ and edge $(i,j)\in\mathcal{E}$ refers to the interconnection between the battery units $i$ and $j$, where $\mathcal{V}$ and $\mathcal{E}$ denote the vertex set and edge set, respectively. The adjacency  matrix $A=[a_{ij}]$ associated with graph $\mathcal{G}$ is defined as $a_{ij}=a_{ji}=0$ if there is no communication link between unit $i$ and unit $j$, otherwise $a_{ij}=a_{ji}=1$. The neighboring set of agent $i$ is represented by $\mathcal{N}_i=\left\{j\in\mathcal{V}:(i,j)\in \mathcal{E} \right\}$. The Laplacian matrix of graph $\mathcal{G}$ is defined as $L=[l_{ij}]$, where $l_{ij}=-a_{ij}$ if $i\neq j$ and $l_{ii}=\sum_{k=1,k\neq i}^N a_{ik}$. Obviously, the Laplacian $L$ is symmetric with eigenvalues $0=\lambda_1(L)<\lambda_2(L)\leq\lambda_3\leq\cdots\leq \lambda_N(L)$. A graph is connected if there exists a path between any two distinct nodes.

\subsection{System description}
Consider a class of heterogeneous BESS with $N$ battery units, described as follows based on the Coulomb counting method,
\begin{equation}\label{SOC+PowerModel}
\begin{array}{l}
s_i(t)=s_i(0)-\frac{1}{C_i}\int_0^ti_i(\tau)d\tau, \\
p_i(t)=V_i(t)i_i(t), i\in \mathcal{N}=\{1,2,\cdots,N\},
\end{array}
\end{equation}
where $s_i(t)$, $s_i(0)$, $C_i$ and $i_i(t)$ are the SoC value, the initial SoC value, the capacity and  the output  current of the $i$-th battery unit, respectively.
$p_i(t)$ is the output power of the $i$-th battery unit, $V_i(t)$ is the output voltage. Note that $p_i<0$ indicates changing and $p_i>0$ indicates dischanging. In general, the output voltage $V_i(t)$ of each unit is assumed to be unchanged during its operation\cite{tan2011design}. Differentiating both sides of the first equation of (\ref{SOC+PowerModel}) yields
\begin{equation}\label{SOCmodel1}
\dot{s}_i(t)=-\frac{1}{C_i}i_i(t), i\in \mathcal{N}.
\end{equation}
Then, from (\ref{SOC+PowerModel}) and (\ref{SOCmodel1}), one has
\begin{equation}\label{SOCmodel2}
\dot{s}_i(t)=-\frac{1}{C_iV_i}p_i(t), i\in \mathcal{N}.
\end{equation}

For each battery unit $i\in \mathcal{N}$, define the battery unit state
\begin{equation}\label{SOCmodel3}
x_i(t)=\left\{ \begin{array}{l}
C_iV_is_i(t), \text{ discharging  mode},\\
C_iV_i(1-s_i(t)), \text{ charging mode}.
\end{array} \right.
\end{equation}

Moreover, the following assumption is made.
\begin{assumption}[\cite{meng2021distributed}]
There exist constants $a_1,a_2>0$ such that
\begin{equation*}\label{assumption1}
a_1\leq x_i(t)\leq a_2, t\geq 0, i\in \mathcal{N}.
\end{equation*}
\end{assumption}

\begin{remark}\label{remark1}
Assumption \ref{remark1} ensures that SoC of battery units varies within an appropriate range to avoid overcharging or overdischarging during the operation of BESSs.
\end{remark}

Then it deduces from (\ref{SOCmodel2}) that
\begin{equation}\label{Dotx}
\dot{x}_i(t)=\left\{ \begin{array}{l}
C_iV_i\dot{s}_i(t)=-p_i(t), \\
-C_iV_i\dot{s}_i(t)=p_i(t).
\end{array} \right.
\end{equation}

Denote the average battery unit state and the average desired power as, respectively,
\begin{equation*}
x_a(t)=\frac{1}{N}\sum_{i=1}^N x_{i}(t),\\
p_a(t)=\frac{1}{N}p^*(t),
\end{equation*}
where $p^*(t)$ is total desired power.

In this paper, the following assumption on $p^*(t)$ is required.
\begin{assumption}[\cite{xing2019distributed,meng2021distributed}]
The total desired  charging/discharging power $p^*(t)$ of  BESSs satisfies
\begin{equation*}\label{assumption2}
\underline {P} \le\mid p^*(t)\mid\le \bar P, t\geq 0,\\
\mid\dot{p}^*(t)\mid \le \epsilon, t\geq 0,
\end{equation*}
\end{assumption}
where $\bar P$, $\underline P$, and $\epsilon$ are positive constants.

In this paper, each battery unit is able to communicate with its neighbors, then the following assumptions on the communication topology of BESSs are made.
\begin{assumption}[\cite{xing2019distributed,meng2021distributed}]\label{assumption3}
The communication topology of BESSs is an undirected connected graph.
\end{assumption}
\begin{assumption}[\cite{meng2021distributed}]\label{assumption4}
At least one battery unit has access to the total desired  power $p^*(t)$.
\end{assumption}

From Assumptions \ref{assumption3} and \ref{assumption4}, the matrix $H=L+B>0$\cite{hu2012robust}, where the diagonal matrix $B = {\rm diag} \{b_1,b_2,\cdots,b_N\}$, $b_i=1$ if the $i$-th battery unit gets access to the total desired power $p^*(t)$ and $b_i=0$ otherwise.

The control objective of BESS is stated as follows.
\begin{problem}\label{problem}
  Consider BESSs (\ref{SOC+PowerModel}), assume that Assumptions\ref{assumption1}-\ref{assumption4} hold. By using the local information,  design the SoC balancing and the charging/discharging power tracking distributed controllers for the battery units so that
\begin{enumerate}
\item all battery units achieve SoC balancing with  any  pre-specified accuracy $\varepsilon_{soc}\ge 0$ , i.e.,
\begin{equation*}
\lim_{t \to \infty} \|s_i(t)-s_j(t)\|\le \varepsilon_{soc}, i,j \in \{1,2,\cdots,N\};
\end{equation*}
\item  the total charging/discharging power of BESSs tracks the total desired power with  any  pre-specified accuracy $\varepsilon_{power}\ge 0$, i.e.,\\
\begin{equation*}
\lim_{t \to \infty} \|\sum\limits_{i=1}^N p_i-p^*(t)\| \le \varepsilon_{power}, i \in \{1,2,\cdots,N\}.
\end{equation*}
\end{enumerate}
\end{problem}

\section{Power Control Based On Prespecified-time observer}

To solve Problem \ref{problem}, the following control strategy is developed in \cite{xing2019distributed},
\begin{equation}\label{DistributedController1}
p_i(t)=\frac{x_i(t)}{x_a(t)}p_a(t).
\end{equation}
Note that both $p_a(t)$ and $x_a(t)$ are global information, which are not available to every battery unit. Recently, \cite{meng2021distributed} introduced two types of distributed observers, one can achieve asymptotic estimation on $p_a(t)$ and $x_a(t)$, while the other ensures finite-time estimation. Unfortunately, the proposed distributed asymptotic observers in \cite{meng2021distributed} can only be carried out as time tends to infinity, while the observer convergence rate of the distributed finite-time observers depends on the initial values of BESSs. In fact, the prespecified-time convergence property is more desirable. Motivated  by \cite{shao2021prespecified}, firstly, the following prespecified-time distributed observer is proposed to estimate $p_a(t)$,
\begin{equation}\label{PTDObserver1}
\begin{array}{l}
\dot{\hat{p}}_{a,i}=-\alpha {\rm sign}(v_i)-\psi\frac{r}{t_b-t_0}\omega(t)v_i,\\
v_i=\sum_{j\in {N_i}}(\hat{p}_{a,i}-\hat{p}_{a,j})+b_i(\hat{p}_{a,i}-p_a),
\end{array}
\end{equation}
where $\hat{p}_{a,i}$ is the $i$-th battery  unit's estimation of $p_a(t)$, $\alpha, \psi$ and $r$ are positive parameters, $v_i$  is the internal state,
the function $\omega(t)$ is defined as follows:
\[\omega(t)=
\begin{cases}
\frac{t_b-t_0}{t_b-t} & t \in \left[t_0,t_b\right),\\
1 & t \in \left[t_b,\infty\right),
\end{cases}
\]
where $t_b$ is a prespecified time.

\begin{lemmax}\label{lemma1}
Consider BESSs (\ref{SOC+PowerModel}) with the prespecified-time distributed observer (\ref{PTDObserver1}). Assume that Assumptions \ref{assumption2}-\ref{assumption4} hold. Then, the average desired power $p_a(t)$ can be estimated accurately at prespecified time $t_b$, i.e., $\hat{p}_{a,i}(t)=p_a(t)$ over $t\in[t_b,+\infty)$, if the observer parameter $\alpha$ is selected such that $\alpha\geq \frac{\epsilon}{N}$.
\end{lemmax}

\begin{proof}
Denote $\hat{p}_{a} = (\hat{p}_{a,1},\hat{p}_{a,2}, \cdots, \hat{p}_{a,N})^T$, $ \tilde{p} =\hat{p}_{a}-1_Np_a,$  and
$ \tilde{p}  = (\tilde{p}_1,\tilde{p}_2, \cdots,\tilde{p}_N)^T$, then it follows from (\ref{PTDObserver1}) that
\begin{align*}
v&=L\hat{p}_a+B\hat{p}_a-B1_Np_a\\
&=H\hat{p}_a-B1_Np_a\\
&=H\tilde p,
\end{align*}
where $H=L+B$.

Consider the Lyapunov function candidate $ V_1(t)=\frac{1}{2}v^T H^{-1}v$. Then, The time derivative of $V_1(t)$ is evaluated as
\begin{align}\label{DotV1}
\dot{V}_1(t)&=v^TH^{-1}H\dot{\tilde p}\notag\\
&=v^T(\dot{\hat{P}}_a-1_N\frac{1}{N}\dot{p}^*)\notag\\
&\leq-\alpha \sum_{i=1}^N |v_i| -\frac{ r\psi}{t_b-t_0}\omega(t) \sum_{i=1}^Nv_i^2+\frac{\epsilon}{N}\|v\|_1\notag\\
&\le(-\alpha+\frac{\epsilon}{N})\|v\|_1- \frac{r \psi}{t_b-t_0}\omega(t)\|v\|_2^2\notag\\
&\le -\frac{r \psi}{t_b-t_0}\omega(t)\|v\|_2^2 \notag\\
&\le -\frac{r \psi}{t_b-t_0}\omega(t)\frac{2}{\lambda_{\max}(H^{-1})}V_1(t),
\end{align}
by using the fact that $\alpha\geq\frac{\epsilon}{N}$ and $\frac{1}{2}v^TH^{-1}v \le \frac{1}{2}\lambda_{\max}(H^{-1})\|v\|_2^2.$

Denote $\Omega(t)=\omega^r(t)$, it gives that
\begin{equation}\label{omega}
\frac{r}{t_b-t_0}\omega(t)=\frac{\dot{\Omega}(t)}{\Omega(t)}.
\end{equation}

According to (\ref{DotV1}) and (\ref{omega}), it gets that
\begin{equation}\label{BoundednessV}
V_1(t)\leq [\omega(t)]^{-\frac{2\psi r}{\lambda_{\max}(H^{-1})}}[\omega(t_0)]^{\frac{2\psi r}{\lambda_{\max}(H^{-1})}}V_1(t_0).
\end{equation}
Moreover, one from the definition of $\omega(t)$ that
\begin{equation}
[\omega(t_0)]^{\frac{2\psi r}{\lambda_{\max}(H^{-1})}}=1,\lim_{t \rightarrow t^{-}_b}[\omega(t)]^{-\frac{2\psi r}{\lambda_{\max}(H^{-1})}}=0,
\end{equation}
which implies that
\begin{equation}
0\leq\lim_{t \rightarrow t^{-}_b}V_1(t)\leq 0.
\end{equation}

Therefore, $\lim_{t \rightarrow t^{-}_b}\| v \|=0$, so does $\lim_{t \rightarrow t^{-}_b}\| \tilde{p}\|=0$, that is, the average desired power $p_a(t)$ is estimated accurately at prespecified time $t_b$.

Furthermore, for $t\in [t_b,+\infty)$, following the similar step in above, one obtains that
\begin{equation*}
\dot{V}_1(t)\le-\frac{r \psi}{t_b-t_0}\|v\|_2^2.
\end{equation*}

Since $v$ is continuous, $V_1(t)$ is also continuous. Therefore, $\lim_{t \rightarrow t_b}V_1(t)=V_1(t_b)$. Then one can get
\begin{equation*}
0\leq V_1(t)\leq V_1(t_b)=0,t\in [t_b,+\infty),
\end{equation*}
which means that $V_1(t)\equiv0$ on $t\in[t_b,+\infty)$, thus $\tilde{p}\equiv0$ and $\dot{\hat{p}}_a \equiv0$ over $t\in[t_b,+\infty)$.

Lastly, the boundedness  of the proposed observer over $[t_0, \infty)$ is checked.

Denote $\tilde{\psi}_1=\frac{\psi }{\lambda_{\max}(H^{-1})}$, for $t\in [t_0,t_b)$, it can be obtained from (\ref{BoundednessV}) that
\begin{equation}
\|\tilde{p}(t)\|\leq \sqrt{\frac{\lambda_{\max}(H)}{\lambda_{\min}(H)}}[\omega(t)]^{-\tilde{\psi}_1 r}\|\tilde{p}(t_0)\|
\end{equation}
and
\begin{equation}
\|\omega(t)\tilde{p}(t)\|\leq \sqrt{\frac{\lambda_{\max}(H)}{\lambda_{\min}(H)}} [\omega(t)]^{-(\tilde{\psi}_1 r-1)}\|\tilde{p}(t_0)\|.
\end{equation}

When $\tilde{\psi}_1 r>2$, one has that $0<[\omega(t)]^{-\tilde{\psi}_1 r}<1$ and $0<[\omega(t)]^{-(\tilde{\psi}_1 r-1)}<1$. Then
\begin{equation}
\max \{\|\omega(t)\tilde{p}(t)\|, \|\tilde{p}(t)\|\}\leq\ \sqrt{\frac{\lambda_{\max}(H)}{\lambda_{\min}(H)}} \|\tilde{p}(t_0)\|.
\end{equation}
Moreover, it follows from (\ref{PTDObserver1}) that
\begin{align}
\|\dot{\hat{p}}_a\|&\leq\ \alpha\|H\|\|\tilde{p}\|+ \frac{r\psi}{t_b-t_0} \|H\|\|\omega(t)\tilde{p}\|\notag\\
&\leq (\alpha+ \frac{r \psi}{t_b-t_0} )\|H\|\sqrt{\frac{\lambda_{\max}(H)}{\lambda_{\min}(H)}}\|\tilde{p}_0\|,
\end{align}
which confirms that $\hat{p}_a$ is bounded on $[t_0,t_b)$. Similarly, it can be seen from the definition that $\hat{p}_a(t)$ is bounded on $[t_b,\infty)$.
\hfill $\blacksquare$
\end{proof}

For each battery unit $i\in \mathcal{N}$, to estimate $x_a(t)$, the following prespecified-time battery unit average state observer is further constructed:
\begin{align}\label{PTDObserver2}
\dot{q}_i=&-\beta {\rm sign}(\hat{\xi}_{i})-\psi\frac{r}{t_a}\omega(t) \hat{\xi}_{i} \notag\\
\hat{\xi}_{i}=&\sum_{j=1}^n  a_{ij}(\hat{x}_{a,i}-\hat{x}_{a,j})\notag\\
\hat{x}_{a,i}=&\sum_{j=1}^n  a_{ij} (q_i-q_j)+x_i
\end{align}
where $\hat{x}_{a,i}$ is the estimate state of the $i$-th battery unit for $x_a(t)$, $\beta$ is a positive parameter, $\hat{\xi}_{i}$ and $q_i$ are the internal states.

\begin{lemmax}\label{lemma2}
Consider BESSs (\ref{SOC+PowerModel}) with the prespecified-time distributed observer (\ref{PTDObserver2}). Assume that Assumptions \ref{assumption1}, \ref{assumption3} and \ref{assumption4} hold. Then, the average battery unit state $x_a(t)$ can be estimated accurately at prespecified time $t_b$, i.e., $\hat{x}_{a,i}(t)=x_a(t)$ over $t\in[t_b,+\infty)$, if the observer parameter $\beta$ is selected such that $\beta\geq \frac{\bar P}{\sqrt N \lambda_2(L)}$.
\end{lemmax}

\begin{proof}
Denote $ \hat{x}_a  = (\hat{x}_{a,1},\hat{x}_{a,2}, \cdots, \hat{x}_{a,N})^T$, $\tilde{x} =\hat{x}_{a}-x_a1_N$ and $\tilde{x} = (\tilde{x}_1,\tilde{x}_2, \cdots,\tilde{x}_N)^T$, then one can obtain that
\begin{align}\label{tildeX}
\dot{\tilde{x}}&=\dot{\hat{x}}_{a}-\dot{x}_a1_N \notag\\
&=L\dot{q}+\dot{x}-\dot{x}_a 1_N\notag\\
&=L\dot{q} +I_n \dot{x} -\frac{1}{n}1_N 1_N^T \dot{x}\notag\\
&= L\dot{q}+LL^\dagger\dot{x},
\end{align}
where  $L^\dagger$ is the generalized inverse matrix of $L$.

Consider the Lyapunov function candidate $V_2(t)=\frac{1}{2}\tilde{x}^T \tilde{x}$. The derivative of $V_2(t)$ along the trajectory of (\ref{tildeX}) is evaluated as
\begin{align}
\dot{V}_2(t)&=\tilde{x}^T\dot{\tilde{x}}=\tilde{x}^T L \dot{q}+\tilde{x}^T LL^\dagger\dot{x}\notag \\
&=\mu^T(t)\dot{q} +\mu^T(t)L^\dagger\dot{x},
\end{align}
where $ \mu (t)=L \tilde{x}$.

In fact,
\begin{align}
\mu^T(t)\dot{q}
&=-\sum_{i=1}^n\quad \beta|\mu_i(t)|-\mu^T(t)\psi\frac{r}{t_b-t_0}\omega(t)L\hat{x}_a\notag\\
&\le-\beta\|\mu(t)\|_1-\psi\frac{r}{t_b-t_0}\omega(t)\tilde{x}^TL^TL(\tilde{x}+x_a1_N)\notag\\
&\le-\beta\|\mu(t)\|_1-\psi\frac{r}{t_b-t_0}\omega(t)\tilde{x}^TL^TL\tilde{x}.
\end{align}
and
\begin{align}
\mu^T(t)L^\dagger\dot{x} &\le\sigma_{\max}(L^\dagger)\|\mu(t)\|_2\|\dot{x}(t)\|_2 \notag\\
&\le\sigma_{\max}(L^\dagger)\sqrt N\|\mu(t)\|_1\|\dot{x}(t)\|_\infty \notag\\
&\le \frac{\sqrt N}{\lambda_2(L)}\|\mu(t)\|_1\|\dot{x}(t)\|_\infty \notag\\
&\le \frac{\sqrt N\bar P}{N\lambda_2(L)}\|\mu(t)\|_1,
\end{align}
where $\sigma_{\max}(\cdot)$ denotes the maximum singular value. Then,  one can further get
\begin{align}
\dot{V}_2(t)&\le(\frac{\sqrt N\bar P}{N\lambda_2(L)}-\beta)\|\mu(t)\|_1-\psi\frac{r}{t_a}\omega(t)\tilde{x}^TL^TL\tilde{x}\notag\\
&\le-\psi\frac{r}{t_a}\omega(t)\tilde{x}^TL^TL\tilde{x},
\end{align}
by using $\beta \geq \frac{\bar P}{\sqrt N \lambda_2(L)}$.

Since the matrix $L$ is a symmetrical, then there exists a matrix $T$ such that $L=T^H \Lambda T$. Let $Y=T^H\tilde{x}$, one can deduce that
\begin{align}
\dot{V}_2(t) &\le-\psi\frac{r}{t_b-t_0}\omega(t)Y^TT^HL^TLTY\notag\\
&\le-\psi\frac{r}{t_b-t_0}\omega(t)\lambda_2(L^TL)Y^TY\notag\\
&\le -2\psi\frac{r}{t_b-t_0}\omega(t)\lambda_2(L^TL)V_2(t)\notag\\
&= -2\psi\frac{\dot{\Omega}(t)}{\Omega(t)}\lambda_{2}(L^TL)V_2(t).
\end{align}
Which means that
\begin{equation}\label{V2}
V_2(t)\le [\omega(t)]^{-2r\psi\lambda_{2}(L^TL)}[\omega(t_0)]^{2r\psi\lambda_{2}(L^TL)}V_2(t_0),
\end{equation}
and
\begin{equation}
0\leq\lim_{t \rightarrow t^{-}_b}V_2(t)\leq0.
\end{equation}
Therefore, $\lim_{t \rightarrow t^{-}_b}\| \tilde{x}\|=0 $, that is, the average battery unit state $x_a(t)$ is estimated accurately at prespecified time $t_b$.
Following the similar step in Lemma \ref{lemma1}, for $t\in [t_b,+\infty)$, one obtains that  $V_2(t)\equiv0$ on $t\in[t_b,+\infty)$, thus $\tilde{x}\equiv0$ over $t\in[t_b,+\infty)$.

Lastly, the boundedness of $\hat{x}_{a,i}$ over $[t_0, \infty)$ is verified.

Denote $\tilde{\psi}_2=\psi\lambda_{2}(L^TL)$, for $t\in [t_0,t_b)$, it can be obtained from (\ref{V2}) that
\begin{equation}
\|\tilde{x}(t)\|\leq[\omega(t)]^{-\tilde{\psi}_2 r}\|\tilde{x}(t_0)\|
\end{equation}
and
\begin{equation}
\|\omega(t)\tilde{x}(t)\|\leq[\omega(t)]^{-(\tilde{\psi}_2 r-1)}\|\tilde{x}(t_0)\|.
\end{equation}

Moreover, it can be derived that
\begin{equation}
\max\{\|\omega(t)\tilde{x}(t), \|\tilde{x}(t)\|\} \|\leq\ \|\tilde{x}(t_0)\|
\end{equation}
when $\tilde{\psi}_2 r>2$.

Therefore, $\|\dot{q}\|\leq\beta +n\frac{r \psi}{t_a} \|L\| \|\tilde{x}(t_0)\|,$ which means that  $q$ is bounded on $[t_0,t_b)$, and
$\|\hat{x}_a\|\leq\  n\|L\|\|q\|+n\|x\|, $ which confirms that $\hat{x}_a$ is bounded on $[t_0,t_b)$. Similarly, it can be seen from the definition that $\hat{x}_a(t)$ is bounded on $[t_b,\infty)$.\hfill $\blacksquare$
\end{proof}

Based on Lemmas \ref{lemma1} and  \ref{lemma2}, the following result on the  performance of the power control algorithms (\ref{DistributedController1})  with the prespecified-time average power observer (\ref{PTDObserver1}) and the prespecified-time battery unit state observer (\ref{PTDObserver2}) is given.

\begin{theoremx}\label{theorem}
Suppose that Assumptions \ref{assumption1}-\ref{assumption4} hold for BESSs (\ref{SOC+PowerModel}), the observer parameters $\alpha$ and $\beta$ are selected as in Lemmas \ref{lemma1} and \ref{lemma2}, respectively, then, under the power control algorithms (\ref{DistributedController1}) with the prespecified-time distributed observers (\ref{PTDObserver1}) and (\ref{PTDObserver2}), it can be concluded that Problem \ref{problem} is solved.
\end{theoremx}

\begin{proof}
Based on Lemmas \ref{lemma1} and \ref{lemma2},  $\hat{p}_{a,i}(t)=p_a(t)$, $\hat{x}_{a,i}(t)=x_a(t)$, when $t\in[t_b,+\infty)$. Then, similar to \cite{xing2019distributed}, both the SoC balancing and the desired power tracking are achieved by the controller (\ref{DistributedController1}). \hfill $\blacksquare$
\end{proof}

\begin{remark}\label{remark2}
In contrast to the asymptotic observers-based cooperative controller or the finite-time observers-based cooperative controller\cite{meng2021distributed}, the proposed cooperative controller relies on the prespecified-time observer, where the time-dependent function $\omega(t)$ is leveraged.
\end{remark}

\begin{remark}\label{remark3}
The function $\omega(t)$ plays a key role in ensuring the prespecified-time stability of the corresponding estimation error systems. Specifically,  from the definition of $\omega(t)$, one has that $\lim_{t \rightarrow t^{-}_b} [\omega(t)]^{-\tau}=0$ and $[\omega(t_0)]^{\tau}=1$ for any $\tau>0$, which leads to $\lim_{t \rightarrow t^{-}_b}V_i(t)=0, i=1,2 $.
\end{remark}

\begin{remark}\label{remark4}
Compared with the asymptotic observers and the finite-time observers presented in \cite{meng2021distributed}, the converge property of the former can only be ensured when $t\rightarrow \infty$, while the converge rate of the latter depends on the initial state of BESSs, the prespecified time observers are designed herein, and its converge rate can be determined in advance according to task requirement. Moreover, in the next section, some simulation examples are given to show the advantage of our results.
\end{remark}

\section{Simulation Examples}
In this section, the effectiveness and superiority of the proposed control strategy with some existing works, such as \cite{meng2021distributed}, are demonstrated by numerical examples.

\begin{example}\label{example1}
Consider the following BESSs with $6$ battery units, the parameters of the battery units are $V=20$V,  $C=220$Ah, and the initial SoC of the battery units are $(0.95,0.86,0.83,0.93,0.97,0.88)$. The communication topology of BESS is shown in Fig. \ref{CommunicationTopology}. Besides, assume that only battery unit $1$ has access to the desired charging/discharging power, i.e., $b_1=1$, and $b_i=0, i=2,3,\cdots, 6$. Moreover, assume that the desired power is given as follows:

\begin{enumerate}
  \item Case 1: $p^*(t)=\pm(4200\sin(t)+4200$)W.
  \item Case 2: \[p^*(t)=\pm
\begin{cases}
1000 & \text{ $t\in[0,0.25)$},\\
1500\sin(5t)+2000 & \text{ $t\in[0.25,0.5)$},\\
5000t & \text{ $t\in[0.5,0.75)$},\\
-4000t+6000 & \text{ $t\in[0.75,1)$} .
\end{cases} \]
\end{enumerate}
\end{example}

\begin{figure}
\centering{\includegraphics[width=5cm]{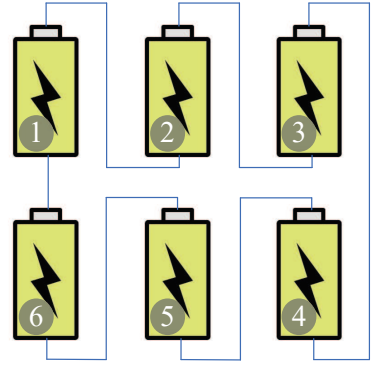}}
\caption{The communication network topology}
\label{CommunicationTopology}
\end{figure}

For the prespecified-time distributed observers (\ref{PTDObserver1}) and (\ref{PTDObserver2}), the prespecified time is chosen as $t_b=4$ and $t_0=0$, the observer parameters are selected as $\psi=4$, $r=50$, $\alpha=1000$, $\beta=3430$, the initial states of the observers (\ref{PTDObserver1}) and (\ref{PTDObserver2}) are given as $v_i(0)=0$, $\hat{p}_i(0)=p_i(0)=0$, $q_i(0)=0$, $\hat{x}_i(0)=x_i(0)$, $i=1,2,\cdots,6$. Then, the simulation results under the controller  (\ref{DistributedController1}) with the observers (\ref{PTDObserver1}) and (\ref{PTDObserver2}) are shown in Figs.\ref{DisCCase1}-\ref{Case2}. Where the solid line represents the simulation results under our control method, while the dotted line represents that in \cite{meng2021distributed}. It can be observed that the proposed cooperative control method is better in comparison with the scheme in \cite{meng2021distributed}.

\begin{figure}[htbp]
\centering
\subfigure[SoC of all battery units.]{
\includegraphics[width=8cm, height=4.5cm]{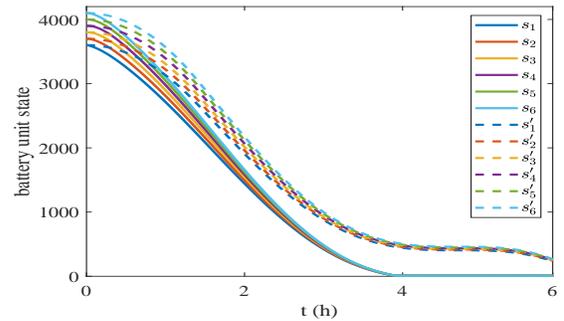}
}
\quad
\subfigure[Output power of all battery units.]{
\includegraphics[width=8cm, height=4.5cm]{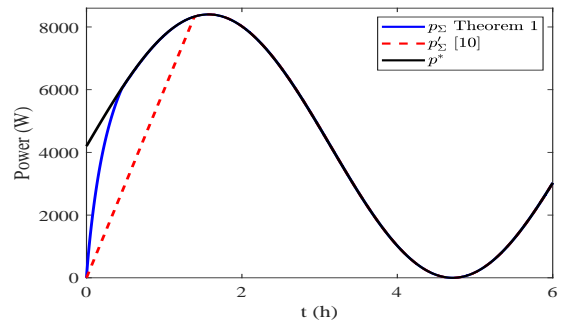}
}
\caption{ BESSs under the controller  (\ref{DistributedController1}) with the observers (\ref{PTDObserver1}) and (\ref{PTDObserver2}) [In discharge mode, Case 1]}
\label{DisCCase1}
\end{figure}

\begin{figure}[htbp]
\centering
\subfigure[SoC of all battery units.]{
\includegraphics[width=8cm, height=4.5cm]{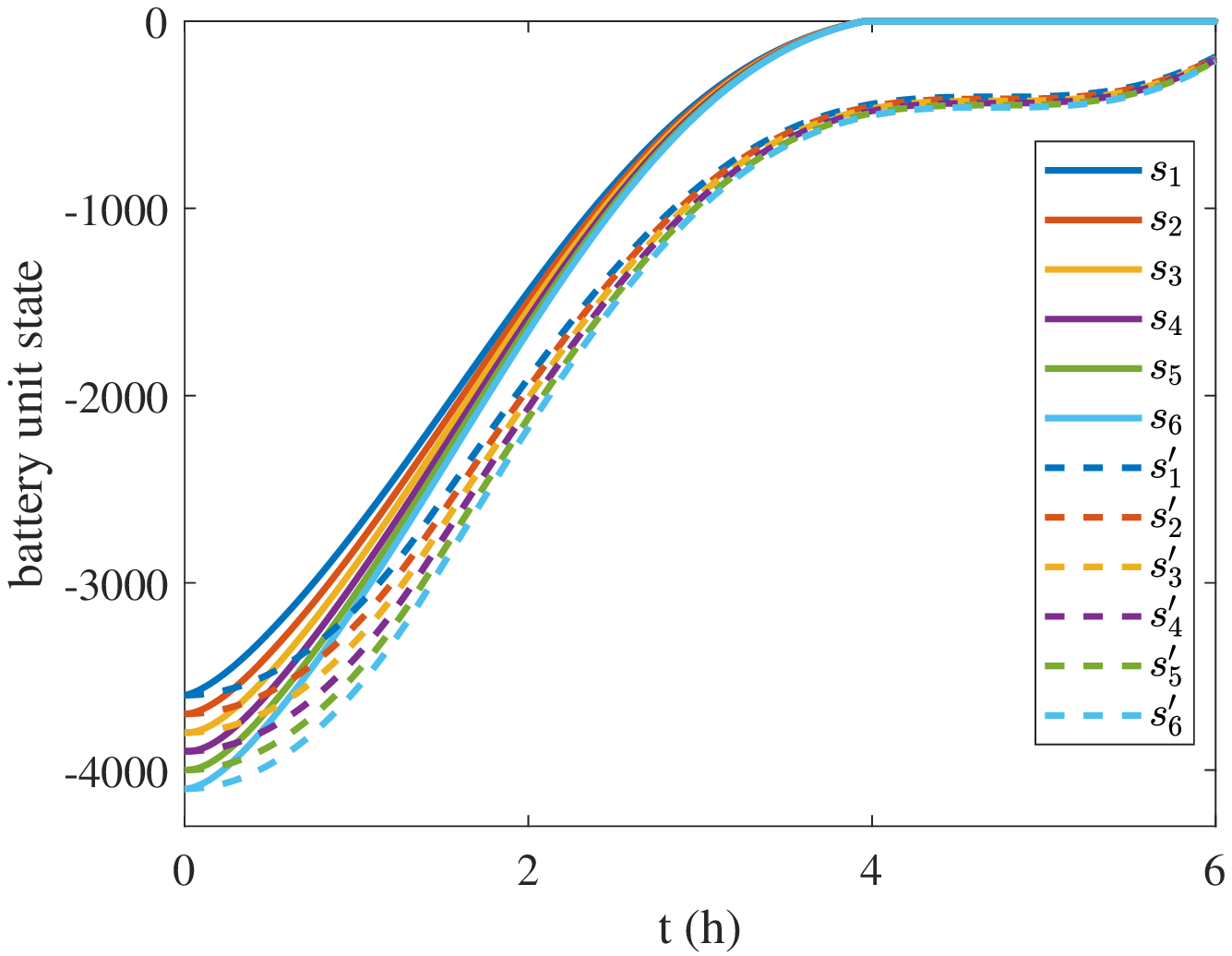}
}
\quad
\subfigure[Output power of all battery units.]{
\includegraphics[width=8cm, height=4.5cm]{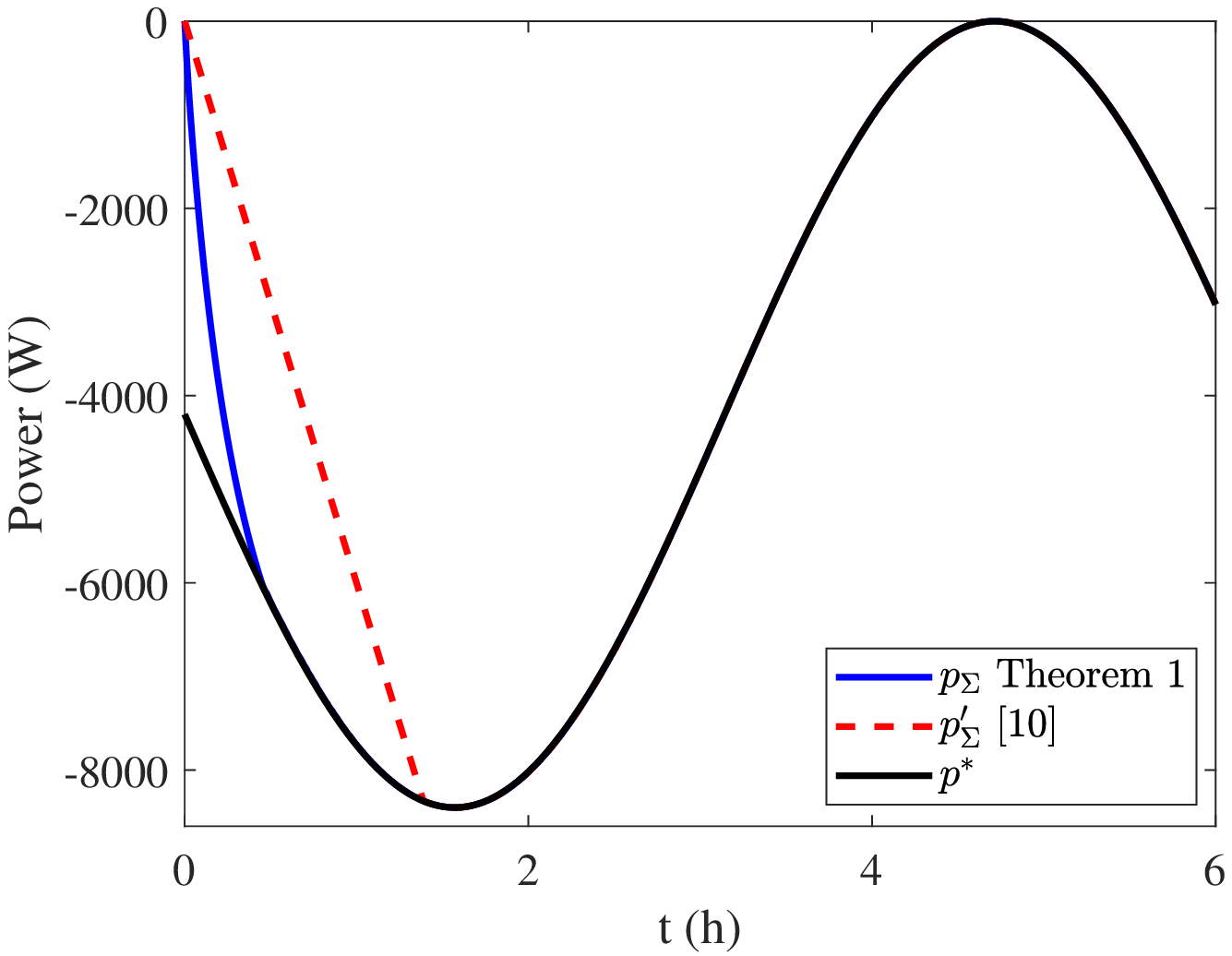}
}
\caption{ BESSs under the controller (\ref{DistributedController1}) with the observers (\ref{PTDObserver1}) and (\ref{PTDObserver2}) [In charge mode, Case 1]}
\label{CCase1}
\end{figure}

\begin{figure}[htbp]
\centering
\subfigure[Output power of all battery units (In discharge mode).]{
\includegraphics[width=8cm, height=4.5cm]{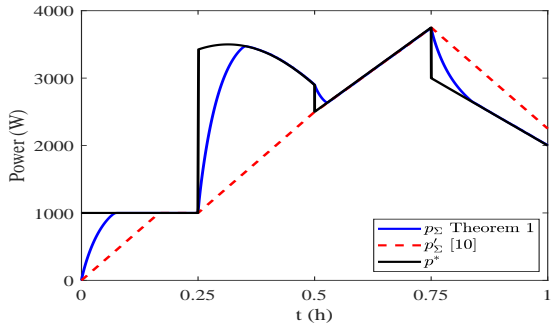}
}
\quad
\subfigure[Output power of all battery units (In charge mode).]{
\includegraphics[width=8cm, height=4.5cm]{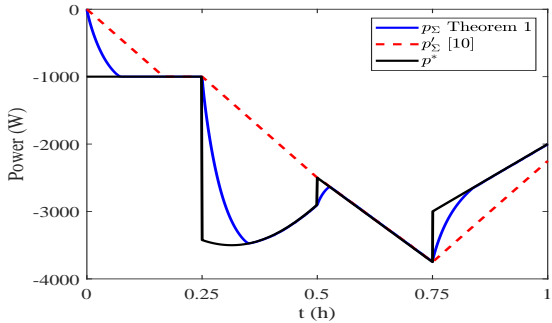}
}
\caption{ BESSs under controllers (\ref{DistributedController1}) with the observers (\ref{PTDObserver1}) and (\ref{PTDObserver2}) [Case 2]}
\label{Case2}
\end{figure}

\section{Conclusions}
This paper has investigated the problem of SoC balancing and power tracking control for BESSs with multiple heterogenous battery units. Two prespecified-time distributed observers have been constructed to estimate the average desired power and the average battery unit state, respectively. Then, by selecting observer's parameters appropriately,
the effectiveness of the proposed control strategy based on the prespecified-time distributed observers has been further analyzed and has been verified by some simulation examples. Future works include extending this results to BESSs with time delays or more communication constraints.


%


\ifCLASSOPTIONcaptionsoff
  \newpage
\fi

\bibliographystyle{ieeetr}        
\bibliography{References}




\end{document}